\documentclass{amsart}
\usepackage{subfigure, graphicx, verbatim, amssymb, float, psfrag,mathrsfs}
\usepackage[mathcal]{euscript}
\usepackage{nicefrac, amsfonts, mathdots}
\usepackage{color}

\newtheorem{theorem}{Theorem}[section]
\newtheorem{lemma}[theorem]{Lemma}
\newtheorem{proposition}[theorem]{Proposition}
\newtheorem{corollary}[theorem]{Corollary}

\theoremstyle{definition}
\newtheorem{definition}[theorem]{Definition}

\theoremstyle{remark}
\newtheorem{remark}[theorem]{Remark}

\numberwithin{equation}{section}

\newcommand\remove[1]{}

\newcommand\nc{\newcommand}
\newcommand\cc{\mathcal{C}}
\newcommand\cd{\mathcal{D}}
\newcommand\ce{\mathcal{E}}
\newcommand\ci{\mathcal{I}}

\newcommand\wt{\mathrm{w}}
\newcommand\rk{\mathrm{rank}}
\newcommand\ip[1]{\langle{#1}\rangle}
\newcommand\rp{\overrightarrow{\mathcal{P}}}
\newcommand\lp{\overleftarrow{\mathcal{P}}}

\newcommand\bo{\mathbf{0}}

\def\supp{\qopname\relax{no}{supp}}
\def\vol{\qopname\relax{no}{vol}}

\def\rank{\qopname\relax{no}{rk}}
\def\shape{\qopname\relax{no}{shape}}
\def\wt{\qopname\relax{no}{w}}

\nc\bfa{{\boldsymbol a}}\nc\bfA{{\bf A}}\nc\cA{{\mathcal A}}\nc\sA{{\mathscr A}}
\nc\bfb{{\boldsymbol b}}\nc\bfB{{\bf B}}\nc\cB{{\mathcal B}}\nc\sB{{\mathscr B}}
\nc\bfc{{\boldsymbol c}}\nc\bfC{{\bf C}}\nc\cC{{\mathcal C}}\nc\sC{{\mathscr C}}
\nc\bfd{{\boldsymbol d}}\nc\bfD{{\bf D}}\nc\cD{{\mathcal D}}\nc\sD{{\mathscr D}}
\nc\bfe{{\boldsymbol e}}\nc\bfE{{\bf E}}\nc\cE{{\mathcal E}}\nc\sE{{\mathscr E}}
\nc\bff{{\boldsymbol f}}\nc\bfF{{\bf F}}\nc\cF{{\mathcal F}}\nc\sF{{\mathscr F}}
\nc\bfg{{\boldsymbol g}}\nc\bfG{{\bf G}}\nc\cG{{\mathcal G}}\nc\sG{{\mathscr G}}
\nc\bfh{{\boldsymbol h}}\nc\bfH{{\bf H}}\nc\cH{{\mathcal H}}\nc\sH{{\mathscr H}}
\nc\bfi{{\boldsymbol i}}\nc\bfI{{\bf I}}\nc\cI{{\mathcal I}}\nc\sI{{\mathscr I}}
\nc\bfj{{\boldsymbol j}}\nc\bfJ{{\bf J}}\nc\cJ{{\mathcal J}}\nc\sJ{{\mathscr J}}
\nc\bfk{{\boldsymbol k}}\nc\bfK{{\bf K}}\nc\cK{{\mathcal K}}\nc\sK{{\mathscr K}}
\nc\bfl{{\boldsymbol l}}\nc\bfL{{\bf L}}\nc\cL{{\mathcal L}}\nc\sL{{\mathscr L}}
\nc\bfm{{\boldsymbol m}}\nc\bfM{{\bf M}}\nc\cM{{\mathcal M}}\nc\sM{{\mathscr M}}
\nc\bfn{{\boldsymbol n}}\nc\bfN{{\bf N}}\nc\cN{{\mathcal N}}\nc\sN{{\mathscr N}}
\nc\bfo{{\boldsymbol o}}\nc\bfO{{\bf O}}\nc\cO{{\mathcal O}}\nc\sO{{\mathscr O}}
\nc\bfp{{\boldsymbol p}}\nc\bfP{{\bf P}}\nc\cP{{\mathcal P}}\nc\sP{{\mathscr P}}
\nc\bfq{{\boldsymbol q}}\nc\bfQ{{\bf Q}}\nc\cQ{{\mathcal Q}}\nc\sQ{{\mathscr Q}}
\nc\bfr{{\boldsymbol r}}\nc\bfR{{\bf R}}\nc\cR{{\mathcal R}}\nc\sR{{\mathscr R}}
\nc\bfs{{\boldsymbol s}}\nc\bfS{{\bf S}}\nc\cS{{\mathcal S}}\nc\sS{{\mathscr S}}
\nc\bft{{\boldsymbol t}}\nc\bfT{{\bf T}}\nc\cT{{\mathcal T}}\nc\sT{{\mathscr T}}
\nc\bfu{{\boldsymbol u}}\nc\bfU{{\bf U}}\nc\cU{{\mathcal U}}\nc\sU{{\mathscr U}}
\nc\bfv{{\boldsymbol v}}\nc\bfV{{\bf V}}\nc\cV{{\mathcal V}}\nc\sV{{\mathscr V}}
\nc\bfw{{\boldsymbol w}}\nc\bfW{{\bf W}}\nc\cW{{\mathcal W}}\nc\sW{{\mathscr W}}
\nc\bfx{{\boldsymbol x}}\nc\bfX{{\bf Z}}\nc\cX{{\mathcal X}}\nc\sX{{\mathscr X}}
\nc\bfy{{\boldsymbol y}}\nc\bfY{{\bf Y}}\nc\cY{{\mathcal Y}}\nc\sY{{\mathscr Y}}
\nc\bfz{{\boldsymbol z}}\nc\bfZ{{\bf Z}}\nc\cZ{{\mathcal Z}}\nc\sZ{{\mathscr Z}}
\nc\od{{\bar d}}\nc\ow{{\bar w}}\nc\odelta{{\bar\delta}}
\nc\ox{{\bar x}}\nc\oy{{\bar y}}\nc\ou{{\bar u}}
\nc\oh{{\bar h}}

\newcommand\ff{{\mathbb F}}

\nc\dgv{\delta_{\text{\rm GV}}}
\nc\dcrit{\delta_{\text{\rm{crit}}}}
\nc\Esp{E_{\text{\rm sp}}}
\renewcommand\epsilon{\varepsilon}

\nc\hr{\overrightarrow{H}}
\nc\hl{\overleftarrow{H}}
\usepackage[colorlinks=true,linkcolor=black,citecolor=black,urlcolor=black]{hyperref}
\begin{document}

\title{Near MDS poset codes and distributions}

\author{Alexander Barg}
\address{Department of ECE/Institute for Systems Research, University
of Maryland, College Park, MD 20817 and 
Dobrushin Mathematical Lab., Institute for Problems of Information
Transmission, Moscow, Russia}
\email{abarg@umd.edu}

\thanks{This research supported in part by NSF grants
DMS0807411, CCF0916919, CCF0830699, and CCF0635271.}

\author{Punarbasu Purkayastha}
\address{Department of ECE/Institute for Systems Research, University
of Maryland, College Park, MD 20817}
\email{ppurka@umd.edu}



\begin{abstract}
We study $q$-ary codes with distance defined by a partial order
of the coordinates of the codewords. Maximum Distance Separable
(MDS) codes in the poset metric have been studied in a number of earlier
works. We consider codes that are close to MDS codes by the value of their
minimum distance. For such
codes, we determine their weight distribution, and in the particular
case of the ``ordered metric'' characterize distributions of points
in the unit cube defined by the codes. We also give some constructions
of codes in the ordered Hamming space.
\end{abstract}

\maketitle

\section{Introduction}

A set of points $C=\{c_1,\dots,c_M\}$ in the $q$-ary $n$-dimensional 
Hamming space $\ff_q^n$ is called a Maximum Distance Separable (MDS) 
code if the Hamming distance 
between any two distinct points of $C$
satisfies $d(c_i,c_j)\ge d$ and the number of points is $M=q^{n-d+1}.$
By the well-known Singleton bound of coding theory, this is the maximum
possible number of points with the given separation. If $C$ is an MDS code 
that forms an $\ff_q$-linear space, then its dimension 
$k$, distance $d$ and length $n$ satisfy the relation $d=n-k+1.$
MDS codes are known to be linked to classical old problems in
finite geometry and to a number of other combinatorial questions related
to the Hamming space \cite{rot06,ald07}.
At the same time, the length of MDS codes cannot be very large;
in particular, in all the known cases, $n\le q+2.$
 This restriction
has led to the study of classes of codes with distance properties close to 
MDS codes, such as $t$-th rank MDS codes \cite{wei91}, 
near MDS codes \cite{dod95} and almost MDS codes \cite{boe96}. 
The distance of these codes is only slightly less than $n-k+1$, 
and at the same time they still have many of the structural properties
associated with MDS codes. 

In this paper we extend the study of linear near MDS (NMDS) codes to 
the case of 
the ordered Hamming space and more generally, to poset metrics.
The ordered Hamming weight was introduced by Niederreiter
\cite{nie86} for the purpose of studying uniform distributions of points 
in the unit cube. The ordered Hamming space in the context of coding theory
was first considered by Rosenbloom and 
Tsfasman \cite{ros97} for a study of one generalization of Reed-Solomon 
codes (the ordered distance is therefore sometimes called the NRT distance).
The ordered Hamming space and the NRT metric  have multiple applications
in coding theory including a 
generalization of the Fourier transform over
finite fields \cite{gun94,mas96}, list decoding of algebraic codes
\cite{nie01b}, and coding for a fading channel of special structure
\cite{ros97,gan07}. This space also gives rise to a range of combinatorial 
problems. 
In the context of algebraic combinatorics, it supports a formally self-dual
association scheme whose eigenvalues form a family of multivariate
discrete orthogonal polynomials \cite{mar99,bie06a,bar09b}.

A particular class of distributions in the unit cube 
$U^n=[0,1)^n$, called $(t,m,n)$-nets,
defined by Niederreiter in the course of his studies, presently forms 
the subject of a large body of literature.
MDS codes in the ordered Hamming space and their relations to 
distributions and $(t,m,n)$-nets
have been extensively studied \cite{ros97,skr01,dou02b,hyu08}.
The ordered Hamming space was further generalized by Brualdi et al. 
in \cite{bru95} which introduced metrics on strings defined by arbitrary
partially ordered sets, calling them poset metrics.

The relation between  MDS and NMDS codes in the ordered metric and
distributions is the main motivation of the present study.
As was observed by Skriganov \cite{skr01}, MDS codes correspond
to optimal uniform distributions of points in the unit cube.
The notion of uniformity is rather intuitive: an 
allocation of $M$ points forms a uniform distribution if every measurable
subset $A\subset U^n$ contains a $\vol(A)$ proportion of the $M$ points
(in distributions that arise from codes, this property is approximated
by requiring that it hold only for some fixed collection of subsets).
Skriganov \cite{skr01} 
observes that distributions that arise from MDS codes are optimal
in some well-defined sense. In the same way, NMDS codes correspond to 
distributions that  are
not far from optimal (they are characterized exactly below).
Although the primary motivation is to study NMDS codes in the ordered
metric, the calculations are easily generalized to the poset metric. We
will hence derive the results in the general case of the poset metric, and
mention the results in the ordered metric as specific cases.  

The rest of the article is organized as follows. In the next section we provide 
basic definitions and some properties of near-MDS codes. We will also have a 
chance to discuss generalized Hamming weights of Wei \cite{wei91} in the 
poset metric case. In Section \ref{sec:nmds-tms} we show a relationship between 
distribution of points in the unit cube and NMDS codes. In Section \ref{sec:wt} 
we determine the weight distribution of NMDS codes, and finally in 
Section \ref{sec:nmds-constr}, we provide some constructions of NMDS codes in 
the ordered Hamming space. 

\section{Definitions and basic properties}
\label{sec:def}

\subsection{Poset metrics}
We begin with defining poset metrics on $q$-ary strings of a fixed length
and introduce the ordered Hamming
 metric as a special case of the general definition. 
Entries of a string $x=(x_1,x_2,\dots)$ are indexed by a finite set
$N$ which we call the set of coordinates.
Let $\rp$ be an arbitrary  partial order ($\le$) on $N.$ 
Together $N$ and $\rp$ form a {\em poset}.
An {\em ideal} of the poset is a subset $I\subset N$ that is
``downward closed'' under the $\le$ relation, which means that the 
conditions $i,j\in N$, $j\in I$ and $i\le j$ imply that $i\in I$. 
For the reasons that will become clear below, such ideals will be called 
{\em left-adjusted} (l.a.).

A chain is a linearly ordered subset of the poset. 
The {\em dual poset} $\lp$ is the set $N$ with the same set of chains
as $\rp$, but the order within each of them reversed. 
In other words $j\le i$ in $\lp$ if and only if $i\le j$ in $\rp$. 
An ideal in the dual poset will be termed {\em right-adjusted} (r.a.).
For a subset $S\subseteq\rp$ we denote by $\ip{S}=\ip{S}_{\rp}$ the smallest 
$\rp$-ideal containing the set $S$  (we write $S\subseteq\rp$ to refer 
to a subset $S\subseteq N$ whose elements are ordered according to $\rp$).
The support of a sequence $x$ is the subset $\supp x\subseteq N$ formed
by the indices of all the nonzero entries of $x.$
The set $\ip{\supp x} \subseteq\rp$ will be called the l.a. support of $x.$
The r.a. support is defined analogously.

\begin{definition} \label{def:poset} (Brualdi et al.~\cite{bru95})
Let $\rp$ be a poset defined on $N$ and let $x,y\in\ff_q^{|N|}$
be two strings.
Define the weight of $x$ with respect to $\rp$ as $\wt(x)=|\ip{\supp x}|,$
i.e., the size of the smallest $\rp$-ideal that contains the support of $x.$
The distance between $x$ and $y$ is defined 
as $d_{\rp} (x,y) = \wt(x-y) = |\ip{\supp (x-y)}|$.
\end{definition}

A code $\cC$ of minimum distance $d$ is a subset of $\ff_q^{|N|}$ such 
that any two distinct vectors $x$ and $y$ of $\cC$ 
satisfy $d_{\rp} (x,y)\ge d.$ It is similarly possible to consider codes 
whose distance is measured relative to $\lp$. In this paper we will be
concerned with {\em linear codes} over a finite field by which we mean linear
subspaces of $\ff_q^{|N|}.$ Given a linear code 
$\cc\subset \ff_q^{|N|}$ its {\em dual code} $\cC^\perp$ is the set of 
vectors $\{y\in\ff_q^{|N|}: \forall_{x\in\cc}\; \sum_i x_iy_i=0\}.$ The weights 
in the dual code $\cC^\perp$ are considered with respect to the dual poset $\lp.$

A subset of $\ff_q^{|N|}$ is called an {\em orthogonal array} of strength
$t$ and index $\theta$ with respect to $\rp$ if any $t$ l.a. columns 
contain any vector $z\in \ff_q^t$ exactly $\theta$ times. In particular,
the dual of a {\em linear} poset code is also a {\em linear} 
orthogonal array.

For instance, the 
Hamming metric is defined by the partial order $\rp$ which is a single 
antichain of length $n=|N|$ (no two elements are comparable). 
Accordingly, the distance between
two sequences is given by the number of coordinates in which they differ.
In this case, $\rp=\lp.$

\subsection{Ordered Hamming metric}
The {\em ordered Hamming metric} is defined by a poset $\rp$ which is 
a disjoint union of $n$ chains of equal length $r.$ 
Since we work with this metric in later sections of the paper, 
let us discuss it in more detail.
In this case $N$ is a union of $n$ blocks of length $r$, and it is 
convenient to write a vector (sequence) as
 $x = (x_{11},\dots,x_{1r},\dots,x_{n1},\dots,x_{nr}) \in \ff_q^{r,n}$.
According to Definition \ref{def:poset}, the weight of $x$ is given by
$$
\wt(x) = \sum_{i=1}^n \max(j: x_{ij} \ne 0). 
$$

For a given vector $x$ let $e_i, i=1,\dots,r$ be
the number of $r$-blocks of $x$ whose rightmost nonzero entry is in the
$i$th position counting from the beginning of the block. The $r$-vector
$e=(e_1,\dots,e_r)$ will be called the {\em shape} of $x$. For brevity we
will write
$$
|e|=\sum_i e_i, \quad |e|'=\sum_{i} ie_i, \quad e_0=n-|e|.
$$
For $I=\ip{\supp x}$ we will denote the shape of the ideal $I$ as
shape($I$)$=e$.
By analogy with the properties of ideals in the ordered Hamming space, 
we use the term ``left adjusted'' for ideals in general posets $\rp.$

An $(nr,M,d)$ {\em ordered code} $\cc\subset \ff_q^{r,n}$ is an arbitrary
subset of $M$ vectors in $\ff_q^{r,n}$ such that the ordered distance between
any two distinct vectors in $\cc$ is at least $d$. If $\cc$ is a linear
code of dimension $k$ over $\ff_q$ and minimum ordered distance $d$, 
we will denote it as an $[nr,k,d]$ code. The dual of $\cc$, denoted as
$\cC^\perp$, is defined as $\cC^\perp=\{x\in\ff_q^{r,n}: \,\forall_{c\in\cc}
\; \sum_{i,j} x_{ij} c_{ij}=0\}.$
The distance in $\cC^\perp$ is derived from the dual order 
$\lp$, i.e., from the r.a. ideals. 

The notion of orthogonal arrays in the ordered Hamming space is derived
from the general definition. They will be called {\em ordered orthogonal
arrays} (OOAs) below. We write $(t,n,r,q)$ OOA for an orthogonal array
 of strength $t$ in $\ff_q^{r,n}.$
Combinatorics of the ordered Hamming space and the duality between codes
and OOAs
was studied in detail by Martin and Stinson \cite{mar99}, Skriganov
\cite{skr01}, and 
the present authors \cite{bar09b}.

\subsection{NMDS poset codes}
We begin our study of NMDS codes in the poset space with several
definitions that are generalized directly from the corresponding
definitions in the Hamming
space \cite{wei91,dod95}. The $t$-th {\em generalized poset weight} of
a linear $[n,k]$ code $\cc$ is defined as 
$$
d_t(\cc) \triangleq \min\{|\ip{\supp \cd}|: \cd \mbox{ is an $[n,t]$
    subcode of }\cc\}, 
$$
where $\supp \cd$ is the union of the supports of all the vectors in $\cd$.
Note that $d_1(\cc) = d,$ the minimum distance of the code $\cc.$
Generalized poset weights
have properties analogous to the well-known
set of properties of generalized Hamming weights.
\begin{lemma}
\label{lemma:dr}
 Let $\cc$ be a linear $[n,k]$ poset code in $\ff_q^n$. Then

     $(1)$ $0< d_1(\cc) < d_2(\cc) < \cdots < d_k(\cc) \le n$.

    $(2)$  Generalized Singleton bound: $d_t(\cc)\le n-\dim(\cc)+t,\quad\forall t\ge1.$

    $(3)$ If $\cC^\perp$ is the dual code of $\cc$ then 
    $$
        \{d_1(\cc),d_2(\cc),\dots,d_k(\cc)\}\cup
   (n+1-\{d_1(\cc^\perp),d_2(\cc^\perp),
  \dots,d_{n-k}(\cc^\perp)\})= \{1,\dots,n\}.
    $$

    $(4)$ Let $H$ be the parity check matrix of $\cc$. Then $d_t(\cc)
        = \delta$ if and only if
        \begin{enumerate}
            \item[(a)] Every $\delta - 1$ l.a. columns of $H$ have
            rank at least $\delta-t$.
            \item[(b)] There exist $\delta$ l.a. columns of $H$ with
            rank exactly $\delta-t$.
        \end{enumerate}
\end{lemma}
\begin{proof}
    (1) Let $\cd_t \subseteq \cc$ be a linear subspace such that
        $|\ip{\supp \cd_t}| = d_t(\cc)$  and $\rk (\cd_t)= t,\ t\ge1$. 
        Let $\Omega(\cd_t)$ denote the maximal elements of the
        ideal $\ip{\supp \cd_t}$. For each coordinate
        in $\Omega(\cd_t)$, $\cd_t$ has at least one vector with a nonzero
        component in that coordinate. We pick $i\in\Omega(\cd_t)$ and let
        $\cd_t^i$ be obtained by retaining only those vectors $v$ in $\cd_t$
        which have $v_i = 0.$ Then 
        $$
            d_{t-1}(\cc) \le | \ip{\supp \cd_t^i }| \le d_t(\cc) - 1.
        $$

(2) This is a consequence of the fact that $d_{t+1} \ge d_t + 1$ and
        $d_k \le n.$

    (3) This proof is analogous to \cite{wei91}. The reason for giving
it here is to assure oneself that no complications arise from the fact that
the weights in $\cc^\perp$ are measured with respect to the dual poset.

We show that for any $1\le s\le n-k-1$, 
         $$
        n+1-d_s(\cC^\perp) \notin \{d_r(\cc): 1\le r\le k\}.
         $$ 
Let $t=k+s-d_s(\cc^\perp).$ We consider two cases (one of which can be void),
namely, $r\le t$ and $r\ge t+1$ and show that for each of them,
$n+1-d_s(\cC^\perp)\ne d_r(\cc).$

        Take an $s$-dimensional subcode $\cd_s\subseteq \cC^\perp$
        such that  $|\ip{\supp \cd_s}_{\lp}|=d_s(\cC^\perp).$          
        Form a parity-check matrix of the code $\cC$ whose first rows
        are some $s$ linearly independent vectors from $\cd_s$.
        Let $D$ be the complement of $\ip{\supp \cd_s}$ in the set of
        coordinates. Let the submatrix of $H$ formed of all the columns
        in $D$ be denoted by $H[D]$. The rank of $H[D]$ is at most $n-k-s$
        and its corank (dimension of the null space) is at least 
        $$
            |D| - (n-k-s) = n-d_s(\cC^\perp) - n + k + s
                            = k+s - d_s(\cC^\perp).
        $$
        Then $d_t(\cc) \le |D|=n - d_s(\cC^\perp)$ 
and so $d_r(\cc) \le n-d_s(\cC^\perp),$ $1\le r \le t$.
        
        Now let us show that $d_{t+i}(\cc) \ne n+1-d_s(\cC^\perp)$ for all
        $1\le i\le k-t.$ 
Assume the contrary and consider a generator 
   matrix $G$ of $\cc$ with the first $t+i$ rows corresponding to the subcode
        $\cd_{t+i}\subseteq \cc$ with $|\ip{\supp \cd_{t+i}}_{\rp}| = d_{t+i}(\cc)$.
        Let $D$ be the complement of $\ip{\supp \cd_{t+i}}$ in the set of
        coordinates. Then $G[D]$ is a $k\times(n-d_{t+i}(\cC))$ matrix of
        rank $k-t-i.$
        By part (2) of the lemma, $n-d_{t+i}(\cC)\ge k-t-i,$ so
           \begin{align*}
           \dim\ker(G[D])&\ge n-d_{t+i}(\cC)-k+t+i\\
           &=s+i-(d_s(\cc^\perp)+n-d_{t+i}(\cc))\\
           &= s+i-1,
           \end{align*}
where the first equality follows on substituting the value of $k$ 
and the second one by using the assumption.
        Hence $d_{s+i-1}(\cC^\perp) \le |D| = d_s(\cC^\perp)-1,$
         which contradicts part (1) of the lemma.

    (4) Follows by standard linear-algebraic arguments.
\end{proof}
    
\remove{If any $\delta - 1$ l.a. columns of $H$  have rank at
        least $\delta-t$ then the null space of $H$ restricted to the
        $\delta-1 $ l.a. columns will have rank at most $t-1$.
        Hence $d_t(\cc) > \delta-1$. There exist $
        \delta$ l.a. columns with rank $\delta-t$ which implies
        that the corresponding null space of those columns have rank $t$
        and so $d_t(\cc) \le \delta$. 

        Conversely, if $d_t(\cc) = \delta$ then all the codewords in $D_t$
        belong to the null space of $H$. Hence, the matrix $H$ restricted to
        these $\delta $ columns has
        rank $\delta-t$. Finally, any subspace of $\cc$ supported on
        $\delta-1$ coordinates has rank at most $t-1$. The corresponding
        $\delta-1$ l.a. columns of $H$ have rank at least
        $\delta-t$.}


\begin{definition}
A linear code $\cc[n,k,d]$ is called 
NMDS if $d(\cc) = n-k$ and $d_2(\cc) = n-k+2.$
\end{definition}

Closely related is the notion of \emph{almost-MDS code} where 
we have only the constraint that $d(\cc)=n-k$ and there is no constraint on
$d_2(\cc)$. In this work, we focus only on NMDS codes.
The next set of properties of NMDS codes can be readily obtained as
generalizations of the corresponding properties of NMDS codes in the
Hamming space \cite{dod95}.
\begin{lemma} \label{lem:nmds}
Let $\cc \subseteq \ff_q^n$ be a linear
$[n,k,d]$ code in the poset $\rp$.
\begin{enumerate}
    \item $\cc$ is NMDS if and only if
    \begin{enumerate}
        \item Any $n-k-1$ l.a. columns of the parity check matrix
        $H$ are linearly independent.
        \item There exist $n-k$ l.a. linearly dependent columns
        of $H$.
        \item Any l.a. $n-k+1$ columns of $H$ are full ranked.
    \end{enumerate}
    \item If $\cc$ is NMDS, so is its dual $\cC^\perp$.
    \item $\cc$ is NMDS if and only if $d(\cc) + d(\cC^\perp) = n$.
    \item If $\cc$ is NMDS then there exists an NMDS code with parameters
    $[n-1,k-1,d]$ and an NMDS code with parameters $[n-1,k,d]$.
\end{enumerate}
\end{lemma}

\begin{proof}
(1) Parts (a) and (b) are immediate. Part (c) is obtained from
        Lemma \ref{lemma:dr}.

(2)   From Lemma \ref{lemma:dr} we obtain
        $$
            \{n+1-d_t(\cC^\perp), 1\le t\le n-k\}
            = \{1,\dots,n-k-1,n-k+1\}.
        $$
        Hence $d(\cC^\perp) = k$ and $d_2(\cC^\perp) = k+2$.

(3) Let $d(\cc) + d(\cC^\perp) = n$. Then 
   $$
     d_2(\cc^\perp)\ge d(\cc^\perp)+1=n-d(\cc)+1,
   $$
but then by Lemma \ref{lemma:dr}(3), $d_2(\cc^\perp)\ge n-d(\cc)+2.$ 
Next,
   $$
  n\ge d_{n-k}(\cc^\perp)\ge d_2(\cc^\perp)+n-k-2\ge 2n-k-d,
   $$
which implies that $d\ge n-k.$ This leaves us with the possibilities of
$d=n-k$ or $n-k+1,$ but the latter would imply that 
$d(\cc) + d(\cC^\perp) = n+2,$ so $d=n-k.$
Further, $d_2(\cc)\ge n-d(\cc^\perp)+2=n-k+2,$ as required.
            The converse is immediate.

(4) To get a $[n-1,k-1,d]$ NMDS code, delete a column of the parity
        check matrix $H$ of $\cc$ preserving a set of $n-k$ l.a.
        linearly dependent columns.
        To get a $[n-1,k,d]$ NMDS code, delete a column of the generator
        matrix $G$ of $\cc$ preserving a set of $k+1$ r.a.
        columns which contains $k$ r.a. linearly dependent
        columns. 
\end{proof}
\begin{lemma}
\label{lem:oa-poset}
Let $\cC$ be a linear poset code in $\rp$ with distance $d$ and let
$\cC^\perp$ be its dual code.
Then the matrix $M$ whose rows are the codewords of $\cC^\perp$ forms
an orthogonal array of strength $d-1$ with respect to $\rp.$
\end{lemma}
\begin{proof} 
Follows because 
(1), $\cc^\perp$ is the linear span of the parity-check matrix $H$
of $\cc;$ and
(2), any $d-1$ l.a. columns of $H$ are linearly independent.
\end{proof}

\section{NMDS codes and distributions}
\label{sec:nmds-tms}
In this section we prove a characterization of NMDS poset codes 
and then use this result to establish a relationship between 
NMDS codes in the ordered Hamming space $\ff_q^{r,n}$ and
uniform distributions of points in the unit cube $U^n.$
In our study of NMDS codes in the following sections,
we analyze the properties of the code
simultaneously as a linear code and as a linear orthogonal array.

Define the $I$-neighborhood of a poset code $\cC$ with respect to an ideal $I$
as
    $$
   B_I(\cc) = \bigcup_{c \in \cc} B_I(c),
    $$
where $B_I(x)=\{v \in \ff_q^{n}: {\supp (v-x)} \subseteq I\}.$  
We will say that a linear $k$-dimensional code $\cc$ forms an $I$-{\em tiling} 
if there exists a partition $\cc=\cC_1\cup\dots\cup\cC_{q^{k-1}}$ into
equal parts such that
the $I$-neighborhoods of its parts are disjoint. If in addition
the $I$-neighborhoods form a partition of $\ff_q^n,$  we say $\cc$ forms
a {\em perfect $I$-tiling}. 

\begin{theorem} \label{thm:nmds-equiv}
Let $\cc \subseteq \ff_q^{n}$ be an $[n,k,d]$ linear code in the poset
$\rp$. $\cc$ is NMDS if and only if
\begin{enumerate}
\item For any $I\subset \rp, |I|=n-k+1,$ the code $\cc$ forms a perfect $I$-tiling.
\item There exists an ideal $I\subset \rp, |I|=n-k$ 
with respect to which $\cc$ forms an $I$-tiling. No smaller-sized ideals
with this property exist.
\end{enumerate}
\end{theorem}
\begin{proof} Let $\cc$ be NMDS and let $I$ be an ideal of size $n-k+1$. 
Let $H[I]$ be the submatrix of the parity-check matrix $H$ of $\cc$
obtained from $H$ by deleting all the columns not in $I$.
Since $\rank (H[I])=n-k,$ the space $\ker(H[I])$ is one-dimensional.
Let $\cc_1=\ker(H(I))$ and let $\cc_j$ be the $j$th coset of $\cc_1$ in $\cc,$
$j=2,\dots,q^{k-1}.$ By Lemma \ref{lem:oa-poset} the code $\cc$ forms an orthogonal array of strength
$k-1$ and index $q$ in $\lp.$ Therefore, every vector $z\in \ff_q^{k-1}$
appears exactly $q$ times in the restrictions of the codevectors $c\in \cc$
to the coordinates of $J=I^c.$
Thus, $c'[J]=c''[J]$ for any two vectors 
$c',c''\in \cC_i, i=1,\dots,q^{k-1}$ and $c'[J]\ne c''[J]$
$c'\in \cc_i,c''\in \cc_j, 1\le i<j\le q^{k-1}.$ This implies that 
$\cc$ forms a perfect $I$-tiling, which proves assumption 1 of the theorem.
To prove assumption 2, repeat the above argument taking $I$ to be the support
of a minimum-weight codeword in $\cc.$

To prove the converse, let $I\subseteq\rp,$  
$|I|=n-k+1$ be an ideal and let $\cc_1,\dots,\cc_{q^{k-1}}$ be a partition of 
$\cc$ with
$|\cc_i|=q$ for all $i,$ that forms a perfect $I-$tiling.
This implies that 
$c'[I^c]\ne c''[I^c]$, $c'\in \cc_i,c''\in \cc_j, 1\le i<j\le q^{k-1}.$
In other words, $\cc$ forms an orthogonal array with respect
to $\lp$ of index $q$ and strength $k-1.$ We conclude that $d(\cc^\perp)=k$
or $k+1.$ If it is the latter, then $\cc^\perp$ is MDS with respect to
$\lp$ and so is $\cc$ with respect to $\rp,$ in violation of assumption 2. 
So $d(\cc^\perp)=k$ and $d(\cc)\le n-k.$ 
If the inequality is strict, there exists an ideal $I$ of size $< n-k$
that supports a one-dimensional subcode of $\cc$. Then $\cc$ forms
an $I$-tiling which contradicts assumption 2.

It remains to prove that $d_2(\cc)= n-k+2$.
Assume the contrary, i.e., that there exists a $2$-dimensional subcode
$\cB\subset \cc$ whose l.a. support forms an ideal $I\subset \rp$ of
size $n-k+1.$ The $q^2$ vectors of $\cB$ all have zeros in $I^c$ which
contradicts the fact that $\cc$ forms an orthogonal array of index $q$.
\end{proof}

Next, we use this characterization to relate codes in the ordered
Hamming space $\ff_q^{r,n}$ to distributions.
An idealized uniformly distributed point set $\cc$ would satisfy the
property that
for any measurable subset $A\subset U^n$, 
  $$
    \frac1{|\cc|}\sum_{x\in \cc} 1(x\in A)=\text{vol}(A).
  $$
Distributions that we consider, and in particular $(t,m,n)$-nets,
approximate this property by restricting the subsets $A$ to be boxes
with sides parallel to the coordinate axes.

Let $$\ce \triangleq \left\{E = \prod_{i=1}^n \left[\frac{a_i}{q^{d_i}},
\frac{a_i+1}{q^{d_i}} \right): 0\le a_i<q^{d_i}, 0\le d_i\le r, 1\le i\le
n\right\}$$ 
be a collection of elementary intervals in the unit cube $U^n = [0,1)^n$.
An arbitrary collection
of $q^k$ points in $U^n$ 
is called an \emph{$[nr,k]$ distribution} in the base $q$ (with respect to 
$\ce$). 
A distribution is called \emph{optimal} if every
elementary interval of volume $q^{-k}$ contains exactly one point
\cite{skr01}.
A related notion of $(t,m,n)$-nets, introduced by Niederreiter \cite{nie86}, is obtained if
we remove the upper bound on $l_i$ (i.e., allow that $0\le l_i < \infty$) 
and require that every elementary interval of volume $q^{t-m}$ contain exactly 
$q^t$ points.

An ordered code gives rise to a distribution of points in the unit cube via
the following procedure. A codevector
$(c_{11},\dots,c_{1r},\dots,c_{n1},\dots,c_{nr}) \in \ff_q^{r,n}$ is mapped
to $x=(x_1,\dots,x_n) \in U^n$  by letting 
\begin{equation}
\label{eq:nrt-cube}
x_i = \sum_{j=1}^r c_{ij}
q^{j-r-1}, 1\le i\le n. 
\end{equation}
In particular, an $(m-t,n,r,q)$ OOA of index $q^t$ and size $q^m$ corresponds to
a distribution in which every elementary interval of volume $q^{t-m}$
contains exactly $q^t$ points, and an $(m-t,n,m-t,q)$
OOA of index $q^t$ and size $q^m$ gives rise to a $(t,m,n)$-net \cite{law96a,mul96}.

\begin{proposition} {\rm (Skriganov \cite{skr01})}  An $[nr,k,d]$ MDS code 
in the ordered metric
exists if and only if there exists an optimal $[nr,k]$ distribution.
\end{proposition}
Skriganov \cite{skr07} also considers the concept of {\em nearly-MDS} codes
whose distance asymptotically tends to the distance of MDS codes, and shows
how these codes can give rise to distributions.

The next theorem whose proof is immediate from Theorem \ref{thm:nmds-equiv}
relates ordered NMDS codes and distributions.
\begin{theorem} \label{thm:nmds-cube}
Let $\cc$ be a linear $[nr,k,d]$ code in $\ff_q^{r,n}$ and let
$P(\cc)$ be the corresponding set of points in $U^n$. Then $\cc$ is NMDS
if and only if
\begin{enumerate}
\item Any elementary interval of volume $q^{-(k-1)}$ has exactly $q$
    points of $P(\cc).$
\item There exists an elementary interval $\prod_{i=1}^n
        \left[0,q^{-l_i}\right)$ of volume $q^{-k}$ containing exactly $q$
        points and no smaller elementary intervals of this form 
        containing exactly $q$ points exist.
\end{enumerate}
\end{theorem}

\begin{corollary}
An $[nr,k,d]$ NMDS code $\cc$ in the ordered Hamming space forms a 
$(k-1,n,r,q)$ OOA of index $q$. The corresponding distribution
$P(\cc)\subset U^n$
forms a $(k-r,k,n)$-net for $k-1\ge r$. 
\end{corollary}

\begin{remark} Distributions of points in the unit 
cube obtained from NMDS codes have properties similar to those
of distributions obtained from MDS codes. In particular, the points
obtained from an $[nr,k,d]$ MDS code in $\ff_q^{r,n}$ satisfy part (1) of
Theorem \ref{thm:nmds-cube} and give rise to a $(k-r,k,n)$-net for
$k\ge r$ \cite{skr01}.
\end{remark}

\section{Weight distribution of NMDS codes}
\label{sec:wt}
 Let $\Omega(I)$ be the set of maximal elements of an
ideal $I$ and let $\tilde{I} \triangleq I \setminus \Omega(I)$. 

Let $\cc$ be
an NMDS $[n,k,d]$ linear poset code. Let 
$A_I \triangleq \{x \in \cc: \ip{\supp x} = I \}$ be the number of codewords
with l.a. support exactly $I$ and let $A_s = \sum_{I: |I| = s}
A_I$. 
\begin{theorem}
\label{thm:enum} The weight distribution of $\cc$ has the following form: 
\begin{multline}
A_s = \sum_{I \in \cI_s} \sum_{l=0}^{s-d-1} (-1)^l \binom{|\Omega(I)|}l
(q^{s-d-l}-1) + (-1)^{s-d} \sum_{I\in\cI_s} \sum_{J\in\cI_d(I),
    J\supseteq \tilde{I}} A_J, \\ n\ge s\ge d,
\label{eq:enum}
\end{multline}
where $\cI_s \triangleq \{I \subseteq \rp: |I| = s\}$ and 
$\cI_s(I) \triangleq \{J: J\subseteq I,
|J| = s\}.$
\end{theorem}
{\sc Proof.} 
The computation below is driven by the fact
that ideals are fixed by the sets of their maximal elements.
Additionally, we use the fact that any $k-1$ r.a. coordinates of the code
$\cc$ support an orthogonal array of strength $k-1.$ 

The number of codewords of weight $s$ is given by
$A_s=|\cup_{I\in \cI_s} \cc\cap S_I|,$ where 
$S_I \triangleq \{x\in\ff_q^n: \ip{\supp x} = I\}$ is the sphere with
l.a. support exactly $I$. 
 The above expression can be written as 
$$
\Big|\bigcup_{I\in \cI_s} \cc\cap S_I\Big| = \sum_{I\in\cI_s} 
\Big( \left|\cc \cap B_I^\ast\right|
        - \big| \bigcup_{J\in\cI_{s-1}(I)} \cc\cap B_J^\ast\big|\Big),
$$
where $B_I \triangleq \{x\in\ff_q^n: {\ip{\supp x}}_{\rp}\subseteq I$\}
and $B_I^\ast \triangleq B_I \setminus \mathbf{0}.$
We determine the cardinality of the last term using the inclusion-exclusion
principle. 
\begin{multline}
\label{eq:incl-excl}
\Big| \bigcup_{J\in\cI_{s-1}(I)} \cc\cap B_J^\ast\Big| = \sum_{J\in\cI_{s-1}(I)}
    |\cc\cap B_J^\ast| - \sum_{J_1\ne J_2\in \cI_{s-1}(I)} |\cc\cap B_{J_1}^\ast \cap
    B_{J_2}^\ast| + \cdots \\
        +(-1)^{|\Omega(I)|-1} \sum_{J_1\ne \cdots\ne
        J_{|\Omega(I)|} \in\cI_{s-1}(I)} \bigg| \cc \cap \Big( \bigcap_{i} 
B_{J_i}^\ast \Big) \bigg|.
\end{multline}
Since $\cC^\perp$ has minimum distance $k$, 
$\cc$ forms an orthogonal array of strength $k-1$ with respect to 
the dual poset $\lp$. 
This provides
us with an estimate for each individual term in (\ref{eq:incl-excl}) as
described below. For distinct $J_1,\dots,J_l\in\cI_{s-1}(I)$, 
we let $J\triangleq\cap_{i=1}^l J_i$. 
Using the fact that $J$ does not contain $l$ maximal elements of $I,$ we get 
$$
\Big|\Big\{\{J_1,\dots,J_l\}: J_i \mbox{ distinct}, J_i \in\cI_{s-1}(I),
    i=1,\dots,l\Big\}\Big|
= \binom{|\Omega(I)|}l.
$$
For any $s\ge d+1$ consider the complement $I^c$ of an ideal
$I\in\ci_s$. Since $|I^c|\le n-d-1= k-1,$
the code $\cc$ supports an
orthogonal array of strength $n-s$ and index $q^{s-d}$
in the coordinates defined by $I^c.$
Since $\cap_{i=1}^l B_{J_i}^\ast = B_J^\ast$ and since $B_J^\ast$ does not
contain the $\mathbf{0}$ vector, we obtain
$$
\Big|\cc \cap \big( \bigcap_{i=1}^l B_{J_i}^\ast\big) \Big|
= q^{s-d-l}-1, \quad 1\le l\le s-d-1.
$$
Finally, for $l=s-d$ we obtain $|\cc\cap(\cap_{i=1}^l B_{J_i}^\ast)| = A_J$, and 
$$
\bigg| \bigcup_{J\in\cI_{s-1}(I)} \cc\cap B_J^\ast \bigg| = \sum_{l=1}^{s-d-1}
(-1)^{l-1} \binom{|\Omega(I)|}l(q^{s-d-l}-1)
    + \sum_{J\in\cI_d(I),J\supseteq \tilde{I}} (-1)^{s-d-1} A_J, 
$$
which implies
\begin{multline*}
\sum_{I\in\cI_s} |\cc\cap S_I| = \sum_{I\in\cI_s} \Bigg( (q^{s-d}-1)
- \bigg(
\sum_{l=1}^{s-d-1} (-1)^{l-1} \binom{|\Omega(I)|}l(q^{s-d-l}-1)
    \\+\sum_{J\in\cI_d(I),J\supseteq \tilde{I}} (-1)^{s-d-1} A_J  \bigg)\Bigg). 
\hspace*{.2in}\rule{2mm}{2mm}
\end{multline*}

As a corollary of the above theorem, we obtain the weight distribution of
NMDS codes in the ordered Hamming space $\ff_q^{r,n}$. 
By definition, the number of vectors of ordered weight $s$ in a code $\cc\in \ff_q^{r,n}$
equals 
   $
     A_s=\sum_{e:|e|'=s} A_e,
  $
where $A_e$ is the number of codevectors of shape $e.$

\begin{corollary}\label{cor:woc}
The weight distribution of an ordered NMDS code $\cc\in \ff_q^{r,n}$ is given by
\begin{multline}
\label{eq:enum-nrt}
A_s = \sum_{l=0}^{s-d-1} (-1)^l \left( \sum_{e:|e|'=s} \binom{|e|}l
        \binom{n}{e_0,\dots,e_r}\right) (q^{s-d-l} - 1)+
    \\+ (-1)^{s-d}\sum_{e:|e|'=d} N_s(e) A_e, \quad s=d,d+1,\dots,n,
\end{multline}
where
$$
N_s(e) \triangleq \sum_{f:|f|'=s} \binom{e_{r-1}}{f_r-e_r}
    \binom{e_{r-2}}{(f_r+f_{r-1}) - (e_r
            + e_{r-1})}\cdots\binom{e_0}{|f|-|e|}.
$$
\end{corollary}
\begin{figure}
\centering
\includegraphics[keepaspectratio,scale=0.52]{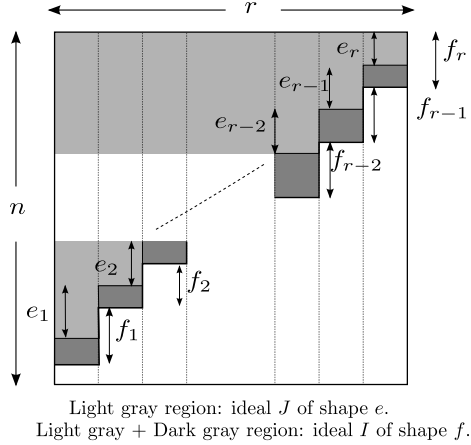}
\caption{To the proof of Corollary \ref{cor:woc}}
\label{fig:ideals}
\end{figure}
\begin{proof} Recall that the shape of an ideal $I$ is 
  $\shape(I)=e=(e_1,\dots,e_r),$ 
where $e_j,j=1,\dots,r$ is the number of chains of length $j$ contained in $I$.
We obtain $|\Omega(I)| = |e|$ and
$$
\sum_{I\in\cI_s} \binom{|\Omega(I)|}l = \sum_{e: |e|' = s} \binom{|e|}l
\binom{n}{e_0,\dots,e_r}.
$$
To determine the last term in (\ref{eq:enum}), we rewrite it as
\begin{align*}
\sum_{I\in\cI_s} \sum_{J\in\cI_d(I), J\supseteq \tilde{I}} A_J &= 
            \sum_{J\in \cI_d} |\{I\in\cI_s: \tilde{I} \subseteq
            J \subseteq I\}| A_J \\
    &= \sum_{e:|e|'=d} N_s(e) \sum_{J:\shape(J) = e} A_J,
\end{align*}
where 
$
N_s(e) = |\{I\in\cI_s: \tilde{I} \subseteq
            J \subseteq I, J \mbox{ fixed, shape}(J) = e\}|.
$

Clearly, $\sum_{J:\shape(J) = e} A_J = A_e,$ and so we only need to
determine the quantity $N_s(e)$ in the above summation. Let $J$ be
an ideal as shown in Fig. \ref{fig:ideals}.
The ideals $I$ which satisfy the constraints in
the set defined by $N_s(e)$ have the form as shown in Fig.
\ref{fig:ideals}.
Letting $f=\shape(I)$, we note that the components of the 
shape $f$ must satisfy
\begin{align*}
    f_r &\ge e_r,\\
    f_r + f_{r-1} &\ge e_r + e_{r-1} \ge f_r, \\
    &\vdots\\
    f_1 +\cdots+ f_r = |f| &\ge |e| = e_1+\cdots+e_r \ge f_2 + \cdots+f_s,\\
    \mbox{and } |f|' &= s.
\end{align*}
It is now readily seen that the cardinality of the set 
   $$
\{I\in\cI_s:
\tilde{I} \subseteq J\subseteq I, J\mbox{ fixed, shape}(J) = e\}
  $$
   is given
by the formula for  $N_s(e)$ as described in (\ref{eq:enum-nrt}). 
\end{proof}

{\bf Remark:} For $r=1$ we obtain $|e| = |e|'=e_1 = d, |f|=f_1 = s $ 
and $N_s(e) = \binom{n-d}{s-d}.$ Thus we recover  the expression for the
weight distribution of an NMDS code in Hamming space \cite{dod95}:
\begin{equation}
\label{eq:enum-ham}
A_s = \sum_{l=0}^{s-d-1} (-1)^l \binom{s}l \binom{n}s(q^{s-d-l}-1)
    + (-1)^{s-d} \binom{n-d}{s-d} A_d.
\end{equation}

Unlike the case of poset MDS codes \cite{hyu08}, 
the weight distribution of NMDS codes is not completely known 
until we know the number of codewords with l.a. support $J$ for every ideal
of weight $J$ of size $d.$
In particular, for NMDS codes in the ordered Hamming space we need to 
know the number of codewords of every shape $e$ with $|e|'=d.$ 
This highlights the fact that the combinatorics of codes in the poset
space (ordered space) is driven by ideals (shapes) and their support
sizes, and that the weight distribution is a derivative invariant of
those more fundamental quantities.

As a final remark we observe that, given that $d(\cc)=n-k$,
the assumption $d(\cc^\perp)=k$ (or the equivalent assumption
$d_2(\cc)=n-k+2$) ensures that the only unknown components of the
weight distribution of $\cc$ correspond to ideals of size $d$.
If instead we consider a code of defect $s,$ 
i.e., a code with $d(\cc)= (n-k+1)-s, \,s\ge 2$,
 it will be possible to compute its weight distribution
using the components $A_J, d\le|J|\le n-d(\cc^\perp)$ (provided
that we know $d(\cc^\perp)$). In the case of the Hamming metric this
was established in \cite{fal98}.

\section{Constructions of NMDS codes}
\label{sec:nmds-constr}
In this section we present some simple constructions of NMDS codes in the
ordered Hamming space for the cases $n=1,2,3.$ We are not aware of any
general code family of NMDS codes for larger $n.$

{\bf n=1:}$\quad$ For $n=1$ the construction is quite immediate once we
recognize that an NMDS $[r,k,d]$ code is also an OOA of r.a. strength $k-1$
and index $q$. Let $I_l$ denote the identity matrix of size $l.$ Let
$x=(x_1,\dots,x_r)$ be any vector of l.a. weight $d=r-k$, i.e. $x_d \ne 0$
and $x_l = 0, \,l=d+1,\dots, r$. Then the following matrix of
size $k\times r$ generates an NMDS code with the above parameters
    \begin{equation}
        \label{eq:constr1}
        \left[\begin{array}{ccc}
        x_1 \dots x_d & 0 & \bo \\
        M & \bo & I_{k-1}
        \end{array}\right],
    \end{equation}
where the $\bo$s are zero vectors (matrices) of appropriate 
dimensions and $M\in\ff_q^{(k-1)\times d}$ is any arbitrary matrix.

{\bf n=2:}$\quad$   
    Let $
    D_l = \left[ \begin{smallmatrix}
    0   & \dots     &   1 \\
\vdots  & \iddots     &   \vdots \\
    1   & \dots     &   0 
    \end{smallmatrix}\right]
    $
    be the $l\times l$ matrix with 1 along the inverse diagonal and
    0 elsewhere. 
    Let $u$ and $v$ be two vectors
    of length $r$ in $\ff_q^{r,1}$ and l.a. weights  $r-k_1$ and
    $r-k_2$ respectively and let $K=k_1+k_2$. The following 
    matrix generates a $[2r,K,2r-K]$ linear NMDS code in $\ff_q^{r,2}$,
    $$
    \left[\begin{array}{cccc|cccc}
    u_1 \,\dots\, u_{r-k_1-1}   & u_{r-k_1} & 0  & \bo       & 
    v_1 \dots v_{r-k_2-1}       & v_{r-k_2} & 0  & \bo       \\
        \bo                     & 0         & 1  & 0         &
        \bo                     & 0         & 1  & 0         \\
        \bo                     & \bo       & \bo&I_{k_1-1} &
     E_r(k_1,k_2)               & \bo       & \bo& \bo       \\
     E_r(k_2,k_1)               & \bo       & \bo& \bo       & 
        \bo                     & \bo       & \bo& I_{k_2-1}
    \end{array}\right],
    $$
    where 
    $E_r(i,j)$ is an $(i-1)\times(r-j-1)$ matrix which has the following
    form:
    $$
    E_r(i,j) = \begin{cases}
        \left[ \begin{array}{c}
        D_{r-j-1} \\ \hline
        \bo_{(i+j-r )\times (r-j-1)}
        \end{array}
        \right],     &   i+j > r, \\
            \\
        \left[ \begin{array}{c|c}
        \bo_{(i-1) \times (r-i-j)} & D_{i-1}
        \end{array}
        \right],     &   i+j \le r.
        \end{cases}
    $$
    From the form of the generator matrix it can be seen that any $K-1$
    r.a. columns of the above matrix are linearly independent.
    But the last $k_1$ and $k_2$ columns from the first 
    and the second blocks respectively are linearly dependent. This
    implies that it forms an OOA of r.a. strength exactly $K-1$. 
    Hence the dual of the code has distance $K$. Finally, the minimum
    weight of any vector produced by this generator matrix is $2r-K$. Hence
    by Lemma \ref{lem:nmds}, this matrix generates
    an NMDS code.

{\bf n=3:}$\quad$ For $n=3$, we have an NMDS code with very specific
        parameters. Let $u, v, w \in \ff_q^{r,1}$ be three vectors of
        l.a. weight $r-2$ each. Then the  matrix shown below is the
        generator matrix of a $[3r,6,d]$ code in base $q\ge3$. 
It is formed of 
three blocks, corresponding to the three dimensions given by $n.$ Here
        $\bo$ is a $1\times(r-6)$ zero vector.
   \begin{align*}
        &\left[ 
        \begin{array}{ccccccc}
            u_1 \dots u_{r-6}&u_{r-5} & u_{r-4} & u_{r-3} & u_{r-2} & 0 & 0     \\
      \bo                    &   0    &     0   &   0     &     0   & 1 & 0     \\
       \bo                   &   0    &     1   &   0     &     0   & 1 & 0     \\
       \bo                   &   1    &     0   &   0     &     0   & 0 & 1     \\
        \bo                  &   0    &     1   &   0     &     0   & 0 & 1     \\
        \bo                  &   0    &     0   &   1     &     0   & 0 & 0     
        \end{array}
        \right]
\\         &\left[
        \begin{array}{ccccccc}
            v_1 \dots v_{r-6}&v_{r-5} & v_{r-4} & v_{r-3} & v_{r-2} & 0 & 0     \\
       \bo                   &   0    &     0   &   0     &     0   & 1 & 0     \\
        \bo                  &   0    &     1   &   0     &     0   & 0 & 0     \\
        \bo                  &   0    &     0   &   1     &     0   & 0 & 0     \\
        \bo                  &   0    &     1   &   0     &     0   & 0 & 1     \\
        \bo                  &   1    &     0   &   0     &     0   & 0 & 1     
        \end{array}
        \right]\\
        &\left[
        \begin{array}{ccccccc}
            w_1 \dots w_{r-6}&w_{r-5} & w_{r-4} & w_{r-3} & w_{r-2} & 0 & 0     \\
        \bo                  &   0    &     1   &   0     &     0   & 0 & 0     \\
        \bo                  &   0    &     0   &   0     &     0   & 1 & 0     \\
       \bo                   &   0    &     1   &   0     &     0   & 0 & 1     \\
      \bo                    &   0    &     0   &   1     &     0   & 0 & 0     \\
      \bo                    &   1    &     0   &   0     &     0   & 0 & 1     
        \end{array}
        \right].
    \end{align*}
        

\providecommand{\bysame}{\leavevmode\hbox to3em{\hrulefill}\thinspace}
\providecommand{\MR}{\relax\ifhmode\unskip\space\fi MR }
\providecommand{\href}[2]{#2}

\end{document}